\RequirePackage{amsmath}
\documentclass[runningheads,envcountsame]{llncs}

\usepackage[T1]{fontenc}
\usepackage{graphicx}
\usepackage{amssymb}
\usepackage[hidelinks, colorlinks=true, linkcolor=blue, citecolor=blue, urlcolor=blue]{hyperref}
\usepackage{cleveref}
\usepackage{enumerate}

\usepackage{amsthm}
\usepackage{thm-restate}

\usepackage{macros}
\begin{document}
	\title{New Width Parameters for Independent Set: One-sided-mim-width and Neighbor-depth\thanks{Tuukka Korhonen was supported by the Research Council of Norway via the project BWCA (grant no. 314528).}}
	\titlerunning{New Width Parameters for Independent Set}
	%
	%
	\author{Benjamin Bergougnoux\inst{1}\orcidID{0000-0002-6270-3663} \and
		Tuukka Korhonen\inst{2}\orcidID{0000-0003-0861-6515} \and
		Igor Razgon\inst{3}}
	\authorrunning{B. Bergougnoux et al.}
	%
	\institute{University of Warsaw, Poland\\\email{benjamin.bergougnoux@mimuw.edu.pl} \and
		University of Bergen, Norway\\\email{tuukka.korhonen@uib.no} \and
		Birkbeck University of London, United Kingdom\\\email{i.razgon@bbk.ac.uk}}
	\maketitle              
	\begin{abstract}
		We study the tractability of the maximum independent set problem from the viewpoint of graph width parameters, with the goal of defining a width parameter that is as general as possible and allows to solve independent set in polynomial-time on graphs where the parameter is bounded.
		We introduce two new graph width parameters: one-sided maximum induced matching-width (o-mim-width) and neighbor-depth.
		O-mim-width is a graph parameter that is more general than the known parameters mim-width and tree-independence number, and we show that independent set and feedback vertex set can be solved in polynomial-time given a decomposition with bounded o-mim-width.
		O-mim-width is the first width parameter that gives a common generalization of chordal graphs and graphs of bounded clique-width in terms of tractability of these problems.

		The parameter o-mim-width, as well as the related parameters mim-width and sim-width, have the limitation that no algorithms are known to compute bounded-width decompositions in polynomial-time.
		To partially resolve this limitation, we introduce the parameter neighbor-depth.
		We show that given a graph of neighbor-depth $k$, independent set can be solved in time $n^{O(k)}$ even without knowing a corresponding decomposition.
		We also show that neighbor-depth is bounded by a polylogarithmic function on the number of vertices on large classes of graphs, including graphs of bounded o-mim-width, and more generally graphs of bounded sim-width, giving a quasipolynomial-time algorithm for independent set on these graph classes.
		This resolves an open problem asked by Kang, Kwon, Str{\o}mme, and Telle [TCS 2017].
		
	\end{abstract}

	\section{Introduction}
	Graph width parameters have been successful tools for dealing with the intractability of NP-hard problems
	over the last decades.
	While tree-width~\cite{DBLP:journals/jct/RobertsonS84} is the most prominent width parameter due to its numerous algorithmic and structural properties, only sparse graphs can have bounded tree-width.
	To capture the tractability of many NP-hard problems on well-structured dense graphs, several graph width parameters, including clique-width~\cite{DBLP:journals/jcss/CourcelleER93}, mim-width~\cite{vatshelle:thesis}, Boolean-width~\cite{DBLP:journals/tcs/Bui-XuanTV11}, tree-independence number~\cite{dallard2022firstpaper,DBLP:conf/soda/Yolov18}, minor-matching hypertree width~\cite{DBLP:conf/soda/Yolov18}, and sim-width~\cite{DBLP:journals/tcs/KangKST17} have been defined.
	A graph parameter can be considered to be more general than another parameter if it is bounded whenever the other parameter is bounded.
	For a particular graph problem, it is natural to look for the most general width parameter so that the problem is tractable on graphs where this parameter is bounded.
	In this paper, we focus on the maximum independent set problem (\textsc{Independent Set}).
	
	Let us recall the standard definitions on branch decompositions.
	Let $V$ be a finite set and $\f  : 2^{V} \rightarrow \mathbb{Z}_{\ge 0}$ a symmetric set function, i.e., for all $X \subseteq V$ it holds that $\f(X) = \f(V \setminus X)$.
	A branch decomposition of $\f$ is a pair $(T, \delta)$, where $T$ is a cubic tree and $\delta$ is a bijection mapping the elements of $V$ to the leaves of $T$.
	Each edge $e$ of $T$ naturally induces a partition $(X_e,Y_e)$ of the leaves of $T$ into two non-empty sets, which gives a partition $(\delta^{-1}(X_e), \delta^{-1}(Y_e))$ of $V$.
	We say that the width of the edge $e$ is $\f(e) = \f(\delta^{-1}(X_e)) = \f(\delta^{-1}(Y_e))$, the width of the branch decomposition $\cD=(T, \delta)$---denoted by $\f(\cD)$---is the maximum width of its edges, and the branchwidth of the function $\f$ is the minimum width of a branch decomposition of $\f$.
	When $G$ is a graph and $\f : 2^{V(G)} \rightarrow \mathbb{Z}_{\ge 0}$ is a symmetric set function on $V(G)$, we say that the \emph{$\f$-width} of $G$ is the branchwidth of $\f$.

	Vatshelle~\cite{vatshelle:thesis} defined the maximum induced matching-width (mim-width) of a graph to be the 
	$\mim$-width where $\mim(A)$ for a set of vertices $A$ is defined to be the size of a maximum induced matching in the bipartite graph $G[A, \overline{A}]$ given by edges between $A$ and $\overline{A}$, where $\overline{A} = V(G) \setminus A$.
	He showed that given a graph together with a branch decomposition of mim-width $k$, any locally checkable vertex subset and vertex partitioning problem (LC-VSVP), including \textsc{Independent Set}, \textsc{Dominating Set}, and \textsc{Graph Coloring} with a constant number of colors, can be solved in time $n^{\OO(k)}$.
	Mim-width has gained a lot of attention recently~\cite{bergougnoux2023logic,DBLP:journals/jgt/BrettellHMPP22,DBLP:journals/ipl/BrettellHMP22,DBLP:journals/tcs/JaffkeKST19,DBLP:journals/dam/JaffkeKT20,DBLP:journals/algorithmica/JaffkeKT20,DBLP:journals/corr/abs-2205-15160}.
	While mim-width is more general than clique-width and bounded mim-width captures many graph classes with unbounded clique-width (e.g. interval graphs), there are many interesting graph classes with unbounded mim-width where \textsc{Independent Set} is known to be solvable in polynomial-time.
	Most notably, chordal graphs, and even their subclass split graphs, have unbounded mim-width, but it is a classical result of Gavril~\cite{DBLP:journals/siamcomp/Gavril72} that \textsc{Independent Set} can be solved in polynomial-time on them.
	More generally, all width parameters in a general class of parameters that contains mim-width and was studied by Eiben, Ganian, Hamm,  Jaffke, and  Kwon~\cite{DBLP:conf/innovations/EibenGHJK22} are unbounded on split graphs.

	With the goal of providing a generalization of mim-width that is bounded on chordal graphs, Kang, Kwon, Str{\o}mme, and Telle~\cite{DBLP:journals/tcs/KangKST17} defined the parameter \emph{special induced matching-width} (sim-width).
	Sim-width of a graph $G$ is the 
	$\sw$-width where $\sw(A)$ for a set of vertices $A$ is defined to be the maximum size of an induced matching in $G$ whose every edge has one endpoint in $A$ and another in $\overline{A}$.
	The key difference of $\mim$ and $\sw$ is that $\mim$ ignores the edges in $G[A]$ and $G[\overline{A}]$ when determining if the matching is induced, while $\sw$ takes them into account, and therefore the sim-width of a graph is always at most its mim-width.
	Chordal graphs have sim-width at most one~\cite{DBLP:journals/tcs/KangKST17}.
	However, it is not known if \textsc{Independent Set} can be solved in polynomial-time on graphs of bounded sim-width, and indeed Kang, Kwon, Str{\o}mme, and Telle asked as an open question if \textsc{Independent Set} is NP-complete on graphs of bounded sim-width~\cite{DBLP:journals/tcs/KangKST17}.

	In this paper, we introduce a width parameter that for the \textsc{Independent Set} problem, captures the best of both worlds of mim-width and sim-width.
	Our parameter is inspired by a parameter introduced by Razgon~\cite{DBLP:conf/iwpec/Razgon21} for classifying the OBDD size of monotone 2-CNFs.
	For a set of vertices $A$, let $E(A)$ denote the edges of the induced subgraph $G[A]$.
	For a set $A \subseteq V(G)$, we define the upper-induced matching number $\umim(A)$ of $A$ to be the maximum size of an induced matching in $G - E(\overline{A})$ whose every edge has one endpoint in $A$ and another in $\overline{A}$.
	Then, we define the \emph{one-sided maximum induced matching-width} (o-mim-width) of a graph to be the
	$\omim$-width where $\omim(A)=\min(\umim(A), \umim(\overline{A}))$.
	In particular, o-mim-width is like sim-width, but we ignore the edges on one side of the cut when determining if a matching is induced.
	Clearly, the o-mim-width of a graph is between its mim-width and sim-width.
	Our first result is that the polynomial-time solvability of \textsc{Independent Set} on graphs of bounded mim-width generalizes to bounded o-mim-width. Moreover, we show that the interest of o-mim-width is not limited to \textsc{Independent Set} by proving that the \textsc{Feedback Vertex Set} problem is also solvable in polynomial time on  graphs of bounded o-mim-width.

	\begin{theorem}\label{thm:omim:algo}
		Given an $n$-vertex graph together with a branch decomposition of o-mim-width $k$, \textsc{Independent Set} and \textsc{Feedback Vertex Set} can be solved in time $n^{\OO(k)}$.
	\end{theorem}

	We also show that o-mim-width is bounded on chordal graphs.
	In fact, we show a stronger result that o-mim-width of any graph is at most its \emph{tree-independence number} (\tin), which is a graph width parameter defined by Dallard, Milani{\v{c}}, and {\v{S}}torgel~\cite{dallard2022firstpaper} and independently by Yolov~\cite{DBLP:conf/soda/Yolov18}, and is known to be at most one on chordal graphs.

	\begin{restatable}{theorem}{omimtin}
		\label{thm:omimtin}
		Any graph with tree-independence number $k$ has o-mim-width at most $k$.
	\end{restatable}

	We do not know if there is a polynomial-time algorithm to compute a branch decomposition of bounded o-mim-width if one exists, and the corresponding question is notoriously open also for both mim-width and sim-width.
	Because of this, it is also open whether \textsc{Independent Set} can be solved in polynomial-time on graphs of bounded mim-width, and more generally on graphs of bounded o-mim-width.

	In our second contribution we partially resolve the issue of not having algorithms for computing branch decompositions with bounded mim-width, o-mim-width, or sim-width.
	We introduce a graph parameter \emph{neighbor-depth}.
	The neighbor-depth $\nd(G)$ of a graph $G$ is defined recursively as follows.
	An empty graph has neighbor-depth $0$, and for a disconnected graph its neighbor-depth is the maximum neighbor-depth of its connected components.
	Then, for a connected non-empty graph $G$, its neighbor-depth is the smallest integer $k$ so that there exists a vertex $v \in V(G)$ so that $\nd(G \setminus N[v]) \le k-1$ and $\nd(G \setminus \{v\}) \le k$, where $N[v] = N(v) \cup \{v\}$ denotes the closed neighborhood of $v$.
	By induction, the neighbor-depth of all graphs is well-defined.
	We show that neighbor-depth can be computed in $n^{\OO(k)}$ time and also \textsc{Independent Set} can be solved in time $n^{\OO(k)}$ on graphs of neighbor-depth $k$.
	
	\begin{theorem}
		\label{thm:neighbordepthmain}
		Given a graph $G$ and an integer $k$, we can decide whether the neighbor-depth is at most $k$ and if so, solve \textsc{Independent Set} in time $n^{\OO(k)}$.
	\end{theorem}

	We show that graphs of bounded sim-width have neighbor-depth bounded by a polylogarithmic function on the number of vertices.
	
	\begin{restatable}{theorem}{simwidthnbdepth}
		\label{thm:simwidthnbdepth}
		Any $n$-vertex graph of sim-width $k$ has neighbor-depth $\OO(k \log^2 n)$.
	\end{restatable}

	Theorems~\ref{thm:neighbordepthmain} and~\ref{thm:simwidthnbdepth} combined show that \textsc{Independent Set} can be solved in time $n^{\OO(k \log^2 n)}$ on graphs of sim-width $k$, which in particular is quasipolynomial time for fixed $k$.
	This resolves, under the mild assumption that $\NP \not\subseteq \QP$, the question of Kang, Kwon, Str{\o}mme, and Telle, who asked if \textsc{Independent Set} is NP-complete on graphs of bounded sim-width~\cite[Question~2]{DBLP:journals/tcs/KangKST17}.

	Neighbor-depth characterizes branching algorithms for \textsc{Independent Set} in the following sense.
	We say that an independent set branching tree of a graph $G$ is a binary tree whose every node is labeled with an induced subgraph of $G$, so that (1) the root is labeled with $G$, (2) every leaf is labeled with the empty graph, and (3) if a non-leaf node is labeled with a graph $G[X]$, then either (a) its children are labeled with the graphs $G[L]$ and $G[R]$ where $(L,R)$ is a partition of $X$ with no edges between $L$ and $R$, or (b) its children are labeled with the graphs $G[X \setminus N[v]]$ and $G[X \setminus \{v\}]$ for some vertex $v \in X$.
	Note that such a tree corresponds naturally to a branching approach for \textsc{Independent Set}, where we branch on a single vertex and solve connected components independently of each other.
	Let $\beta(G)$ denote the smallest number of nodes in an independent set branching tree of a graph $G$.
	Neighbor-depth gives both lower- and upper-bounds for $\beta(G)$.
	
	\begin{theorem}
		\label{thm:nbdepthbranchingtree}
		For all graphs $G$, it holds that $2^{\nd(G)} \le \beta(G) \le n^{\OO(\nd(G))}$.
	\end{theorem}

	By observing that some known algorithms for \textsc{Independent Set} in fact construct independent set branching trees implicitly, we obtain upper bounds for neighbor-depth on some graph classes purely by combining the running times of such algorithms with Theorem~\ref{thm:nbdepthbranchingtree}.
	In particular, for an integer $k$, we say that a graph is $C_{>k}$-free if it does not contain induced cycles of length more than $k$.
	Gartland, Lokshtanov, Pilipczuk, Pilipczuk and Rzazewski~\cite{DBLP:conf/stoc/GartlandLPPR21} showed that \textsc{Independent Set} can be solved in time $n^{\OO(\log^3 n)}$ on $C_{>k}$-free graphs for any fixed $k$, generalizing a result of Gartland and Lokshtanov on $P_k$-free graphs~\cite{DBLP:conf/focs/GartlandL20}.
	By observing that their algorithm is a branching algorithm that (implicitly) constructs an independent set branching tree, it follows from Theorem~\ref{thm:nbdepthbranchingtree} that the neighbor-depth of $C_{>k}$-free graphs is bounded by a polylogarithmic function on the number of vertices.
	
	\begin{proposition}
		\label{pro:ndckfree}
		For every fixed integer $k$, $C_{>k}$-free graphs with $n$ vertices have neighbor-depth at most $\OO(\log^4 n)$.
	\end{proposition}
	
	Along the same lines as Proposition~\ref{pro:ndckfree}, a polylogarithmic upper bound for neighbor-depth could be also given for graphs with bounded induced cycle packing number, using the quasipolynomial algorithm of Bonamy, Bonnet, D{\'{e}}pr{\'{e}}s, Esperet, Geniet, Hilaire, Thomass{\'{e}}, and Wesolek~\cite{DBLP:conf/soda/BonamyBDEGHTW23}.

	\begin{figure}[t!]
		\label{fig:classes}
		\includegraphics[width=\textwidth]{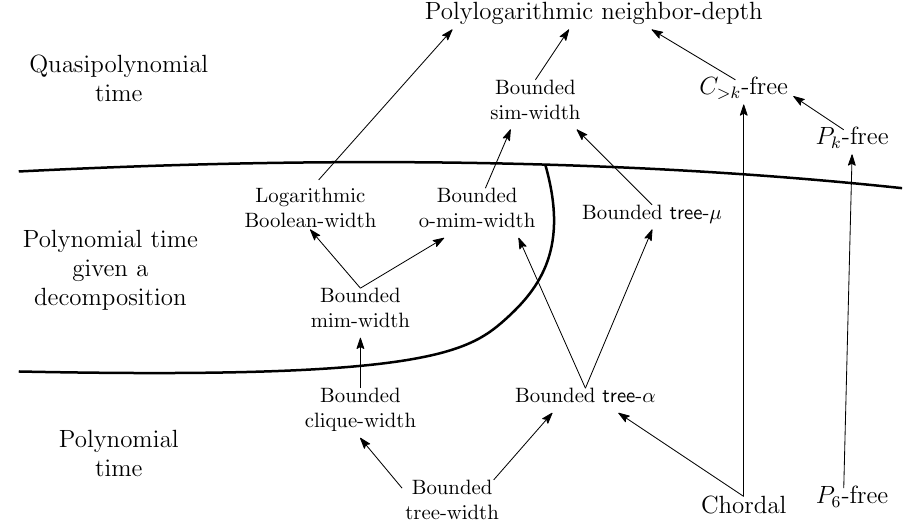}
		\caption{Hierarchy of some graph classes with polylogarithmically bounded neighbor-depth, divided vertically on whether the best known algorithm for \textsc{Independent Set} on the class is polynomial time, polynomial time given a decomposition (and quasipolynomial without a decomposition), or quasipolynomial time. 
		}
	\end{figure}

	In Figure~\ref{fig:classes} we show a hierarchy of graph classes discussed in this paper, and the known algorithmic results for \textsc{Independent Set} on those classes.
	All of the inclusions shown are proper, and all of the inclusions between these classes are included in the figure.
	Some of the inclusions are proven in Sections~\ref{subsec:relomimwidth} and~\ref{subsec:ndsimwidth}, and some of the non-inclusions in Section~\ref{sec:paramrel}.
	Note that bounded Boolean-width is equivalent to bounded clique-width~\cite{vatshelle:thesis}.
	The polynomial-time algorithm for \textsc{Independent Set} on $P_6$-free graphs is from~\cite{DBLP:journals/talg/GrzesikKPP22}, the definition of $\tmm$ and polynomial-time algorithm for \textsc{Independent Set} on graphs of bounded $\tmm$ is from~\cite{DBLP:conf/soda/Yolov18}, and the definition of Boolean-width and a polynomial-time algorithm for \textsc{Independent Set} on graphs of logarithmic Boolean-width is from~\cite{DBLP:journals/tcs/Bui-XuanTV11}.
	The inclusion of logarithmic Boolean-width in polylogarithmic neighbor-depth follows from \Cref{thm:simwidthnbdepth}  and the fact the sim-width of a graph is at most its Boolean-width.
	Polynomial-time algorithm for \textsc{Independent Set} on graphs of bounded clique-width follows from~\cite{DBLP:journals/mst/CourcelleMR00,DBLP:journals/jct/OumS06}.

	\section{Preliminaries}
	\label{sec:prel}
	The size of a set $V$ is denoted by $|V|$ and its power set is denoted by $2^V$. 
	We let  $\max(\emptyset):=-\infty$.
	
	\paragraph{Graphs.}
	Our graph terminology is standard and we refer to \cite{Diestel12}.  The set of vertices of a graph $G$ is denoted by $V(G)$ and the set of edges by $E(G)$.
	For a vertex subset $X\subseteq V(G)$, when the underlying graph $G$ is clear from context, we denote by $\comp{X}$ the set $V(G)\setminus X$.
	An edge between two vertices $x$ and $y$ is denoted by $xy$ or $yx$. 
	The set of vertices that are adjacent to $x$ is denoted by $N_G(x)$. For a set $U\subseteq V(G)$, we define  $N_G(U):=\bigcup_{x\in U}N_G(x) \setminus U$.
	The closed neighborhood of a vertex $x$ is denoted by $N_G[x] = N_G(x) \cup \{x\}$ and the closed neighborhood of a vertex set $U$ by $N_G[U] = N_G(U) \cup U$.
	If the underlying graph is clear, then we may remove $G$ from the subscript.
	
	The subgraph of $G$ induced by a subset $X$ of its vertex set is denoted by $G[X]$.
	We also use the notation $G \setminus X = G[V(G) \setminus X]$.
	For two disjoint subsets of vertices $X$ and $Y$ of $V(G)$, we denote by $G[X,Y]$ the bipartite graph with vertex set $X\cup Y$ and edge set $\{xy \in E(G)\mid x\in X \text{ and } \ y\in Y \}$.
	Given two disjoint set of vertices $X,Y$, we denote by $E(X)$ the set of edges of $G[X]$ and by $E(X,Y)$ the set of edges of $G[X,Y]$. 
	For a set of edges $E'$ of $G$, we denote by $G-E'$ the graph with vertex set $V(G)$ and edge set $E(G)\setminus E'$.
	
	An \emph{independent set} is a set of vertices that induces an edgeless graph.
	Given a graph $G$ with a weight function $w : V(G)\to \bZ_{\ge 0}$, 
	the problem \textsc{Independent Set} asks for an independent set of maximum weight, where the weight of a set $X\subseteq V(G)$ is $\sum_{x\in X} w(x)$.
	A \emph{feedback vertex set} is the complement of a set of vertices inducing a forest (i.e. acyclic graph).
	The problem \textsc{Feedback Vertex Set} asks for a feedback vertex set of minimum weight.
	
	A \emph{matching} in a graph $G$ is a set $M \subseteq E(G)$ of edges having no common endpoint and an \emph{induced matching} is a matching where the subgraph of $G$ induced by the endpoints of the matching does not contain any other edges than the edges of the matching.
	Given two disjoint subsets $A,B$ of $V(G)$, we say that a matching $M$ is a $(A,B)$-matching if every edge of $M$ has one endpoint in $A$ and the other in $B$.

	\paragraph{Width parameters.}
	
	We refer to the introduction for the definitions of branch-decomposition and $\f$-width, we recall below the definitions of mim-width, sim-width and o-mim-width.
	\begin{itemize}
		\item The maximum induced matching-width (mim-width)~\cite{vatshelle:thesis} of a graph $G$ is the $\mim$-width  where $\mim(A)$ is the size of a maximum induced matching of the graph $G[A,\comp{A}]$.

		\item The special induced matching-width (sim-width)~\cite{DBLP:journals/tcs/KangKST17} of a graph $G$ is the $\sw$-width  where $\sw(A)$ is the size of maximum induced $(A,\comp{A})$-matching in the graph $G$.
		
		\item Given a graph $G$ and $A\subseteq V(G)$, the \textit{upper-mim-width} $\umim(A)$ of $A$ is the size of maximum induced $(A,\comp{A})$-matching in the graph $G-E(\comp{A})$.
		The one-sided-mim-width (o-mim-width) of $G$ is the $\omim$-width where $\omim(A)\defeq \min( \umim(A), \umim(\comp{A}))$.
	\end{itemize}
	The following is a standard lemma that $\f$-width at most $k$ implies balanced cuts with $\f$-width at most $k$. 
	
	\begin{lemma}
		\label{lem:balsep}
		Let $G$ be a graph, $X \subseteq V(G)$ a set of vertices with $|X| \ge 2$, and $\f : 2^{V(G)} \rightarrow \bZ_{\ge 0}$ a symmetric set function.
		If the $\f$-width of $G$ is at most $k$, then there exists a bipartition $(A,\overline{A})$ of $V(G)$ with $\f(A) \le k$, $|X \cap A| \le \frac{2}{3}|X|$, and $|X \cap \overline{A}| \le \frac{2}{3}|X|$.
	\end{lemma}
	\begin{proof}
		Let $(T, \delta)$ be a branch decomposition of $\f$ of width $k$, and let us subdivide some edge of $T$ and consider $T$ be rooted on this subdivision node $r$.
		Now, for a node $x$ of $T$, denote by $V_x$ the vertices of $G$ that are mapped to the leafs of $T$ that are descendants of $x$.
		We walk from the root $r$ as follows.
		We start by setting $t$ as the root node, and then while $|V_t \cap X| \ge \frac{2}{3}|X|$, we move $t$ to a child $c$ of $t$ with $|V_c \cap X| \ge \frac{1}{3}|X|$ (note that such a child $c$ exists because the tree is binary).
		This walk must end up in a node $t$ with $\frac{1}{3} |X| \le |V_t \cap X| \le \frac{2}{3}|X|$, giving the desired bipartition $(A,\overline{A}) = (V_t, \overline{V_t})$.
	\end{proof}

	A tree decomposition of a graph $G$ is a pair $(T, \bag)$, where $T$ is a tree and $\bag : V(T) \rightarrow 2^{V(G)}$ is a function from the nodes of $T$ to subsets of vertices of $G$ called \emph{bags}, satisfying that (1) for every edge $uv \in E(G)$ there exists a node $t \in V(T)$ so that $\{u,v\} \subseteq \bag(t)$, and (2) for every vertex $v \in V(G)$, the set of nodes $\{t \in V(T) \colon v \in \bag(t)\}$ induces a non-empty and connected subtree of $T$.
	The width of a tree decomposition is the maximum size of $\bag(t)$ minus one, and the treewidth of a graph is the minimum width of a tree decomposition of the graph.

	For a set of vertices $X \subseteq V(G)$, we denote by $\alpha(X)$ the maximum size of an independent set in $X$.
	The independence number of a tree decomposition $(T,\bag)$ is
	the maximum of $\alpha(\bag(t))$ over $t \in V(T)$ and it is denoted by $\alpha(T,\bag)$.
	The tree-independence number of a graph ($\tin$) is the minimum independence number of a tree decomposition of the graph~\cite{dallard2022firstpaper,DBLP:conf/soda/Yolov18}.
	
	For a set of vertices $X \subseteq V(G)$, we denote by $\mu(X)$ the maximum size of an induced matching in $G$ so that for each edge of the matching, at least one of the endpoints of the edge is in $X$.
	For a tree decomposition $(T,\bag)$, we denote by $\mu(T,\bag)$ the maximum of $\mu(\bag(t))$ over $t \in V(T)$.
	Yolov~\cite{DBLP:conf/soda/Yolov18} defined the minor-matching hypertree width ($\tmm$) of a graph to be the minimum $\mu(T,\bag)$ of a tree decomposition $(T,\bag)$ of $G$.

	\section{O-mim-width}
	\label{sec:omim}
	
	In this section, we prove \Cref{thm:omim:algo,thm:omimtin}.
	We start with some intermediary results.
	The following reveals an important property of cuts of bounded upper-mim-width. Razgon proved a similar statement in~\cite{Razgon21}. 
	To simplify the statements of this section, we fix an $n$-vertex graph $G$ with a weight function $w : V(G) \to \bZ_{\ge 0}$.
	
	\begin{lemma}\label{lem:umim}
		Let $A\subseteq V(G)$. For every $X\subseteq A$ that is the union of $t$ independent sets, there exists $X'\subseteq X$ of size at most $t\cdot \umim(A)$ such that $N(X)\setminus A = N(X')\setminus A$. In particular, we have $\abs{\{N(X)\setminus A \mid X\in \mathsf{IS}(A) \} } \le n^{\umim(A)}$ where $\mathsf{IS}(A)$ is the set of independent sets of $G[A]$.
	\end{lemma}
	\begin{proof}
		It is sufficient to prove the lemma for $t=1$, since if $X$ is the union of $t$ independent sets $X_1,\dots,X_t$, then the case $t=1$ implies that, for each $i\in[1,t]$, there exits $X_i'\subseteq X_i$ such that
		$N(X_i)\setminus A = N(X'_i)\setminus A$ and $\abs{X_i'}\leq \umim(A)$.
		It follows that $X'=X_1'\cup \dots \cup X_t' \subseteq X$, $N(X)\setminus A = N(X')\setminus A$ and $\abs{X'}\leq t\cdot \umim(A)$.
		
		Let $X$ be an independent set of $G[A]$.
		If for every vertex $x\in X$, there exists a vertex $y_x\in \comp{A}$ such that $N(y_x)\cap X= \{x\}$, then $\{xy_x \mid x\in X\}$ is an induced $(A,\comp{A})$-matching in $G - E(\comp{A})$.
		We deduce that either $\abs{X}\le \umim(A)$ or there exists a vertex $x\in X$ such that $N(X)\setminus A = N(X\setminus \{x\}) \setminus A$.
		Thus, we can recursively remove vertices from $X$ to find a set $X'\subseteq X$ of size at most $\umim(A)$ and such that $N(X)\setminus A = N(X')\setminus A$.
		In particular, the latter implies that $\{N(X)\setminus A \mid X\in \mathsf{IS}(A) \} = \{N(X)\setminus A \mid X\in \mathsf{IS}(A) \land \abs{X}\le \umim(A) \}$. 
		We conclude that $\abs{\{N(X)\setminus A \mid X\in \mathsf{IS}(A) \}}\le n^{\umim(A)}$.
	\end{proof}
	
	To solve \textsc{Independent Set} and \textsc{Feedback Vertex Set}, we use the general toolkit developed in \cite{bergougnoux2023logic} with a simplified notation adapted to our two problems.
	This general toolkit is based on the following notion of representativity between sets of partial solutions. 
	In the following, the collection $\cS$ represents the set of solutions, in our setting $\cS$ consists of either all the independent sets or all the set of vertices inducing a forest.
	
	\begin{definition}\label{def:represents}
		Given $\cS\subseteq 2^{V(G)}$, for every $\cA\subseteq 2^{V(G)}$ and $Y\subseteq V(G)$, we define 
		$ 		 \best_{\cS}(\cA,Y) \defeq \max\{w(X) \mid X\in \cA \land X\cup Y \in \cS\}. $		
		Given $A\subseteq V(G)$ and $\cA,\cB\subseteq 2^{A}$, we say that $\cB$ $(\cS,A)$-represents $\cA$ if for every $Y\subseteq \comp{A}$, we have $\best_{\cS}(\cA,Y)=\best_{\cS}(\cB,Y)$.
	\end{definition}
	
	Observe that if there is no $X\in \cB$ such that $X\cup Y\in\cS$, then $\best_{\cS}(\cB,Y)=\max(\emptyset)=-\infty$.
	It is easy to see that the relation ``$(\cS,A)$-represents'' is an equivalence relation.
	
	The following is simplification of Theorem~4.5 from \cite{bergougnoux2023logic}.
	It proves that a routine for computing small representative sets can be used to design a dynamic programming algorithm.

	\begin{theorem}[\cite{bergougnoux2023logic}]\label{thm:reduce}
		Let $\cS\subseteq 2^{V(G)}$. Assume that there exists a constant $c$ and an algorithm that, given $A\subseteq V(G)$ and $\cA\subseteq 2^{A}$, computes in time $\abs{\cA}n^{\OO(\omim(A))}$ a subset $\cB$ of $\cA$ such that $\abs{\cB} \leq n^{c \cdot \omim(A)}$ and $\cB$ $(\cS,A)$-represents $\cA$.
		Then, there exists an algorithm, that given a layout $\cL$ of $G$, computes in time $n^{\OO(\omim(\cL))}$ a set of size at most $n^{c \cdot \omim(A)}$ that contains an element in $\cS$ of maximum weight.  
	\end{theorem}
	
	In the rest of this section, we prove that routines for computing small representative sets exist for \textsc{Independent Set} and \textsc{Feedback Vertex Set}. To simplify the following statements, we fix a subset $A\subseteq V(G)$.
	
	\subsection{Independent Set}
	
	\def\IS{\mathcal{I}}
	
	The following lemma provides a routine to compute small representative sets for \textsc{Independent Set}.
	We denote by $\IS$ the set of all independent sets of $G$.
	
	\begin{lemma}\label{lem:omim:IS}
		Let $k=\omim(A)$.
		Given a collection $\cA\subseteq 2^{A}$, we can compute in time $\abs{\cA} n^{\OO(k)}$ a subset $\cB$ of $\cA$ such that $\cB$ $(\IS,A)$-represents $\cA$ and $\abs{\cB}\le n^{k}$.
	\end{lemma}
	\begin{proof}
		Let $\cA\subseteq 2^{A}$.
		We compute $\cB$ from the empty set as follows:
		\begin{itemize}
			\item If $\umim(A)=k$, then, for every  $Y\in \{N(X)\setminus A \mid X $ is an independent in $\cA\}$, we add to $\cB$ an independent set $X\in \cA$ of maximum weight such that $Y=N(X)\setminus A$.
			
			\item If $\umim(A)> k$, then, for each subset $Y\subseteq \comp{A}$ with $\abs{Y}\le k$, we add to $\cB$ a set $X\in \cA$ of maximum weight such that $X\cup Y$ is an independent set (if such $X$ exists).
		\end{itemize}
		
		\paragraph{Correctness.} 
		First, we prove that $\abs{\cB}\le n^{k}$.
		This is straightforward when $\umim(A)> k$.
		When $\umim(A)=k$, \Cref{lem:umim} implies that $\abs{\{N(X)\setminus A \mid X $ is an independent in $\cA\}}\le n^{k}$ and thus, we have $\abs{\cB}\le n^{k}$.

		Next, we prove that $\cB$ $(\IS,A)$-represents $\cA$, i.e. for every $Y\subseteq \comp{A}$, we have $\best_\IS(\cA,Y)=\best_\IS(\cB,Y)$.
		Let $Y\subseteq \comp{A}$. 
		As $\cB$ is subset of $\cA$, we have $\best_\IS(\cB,Y) \le \best_\IS(\cA,Y)$.
		In particular, if there is no $X\in\cA$ such that $X\cup Y$ is an independent set, then we have $\best_\IS(\cA,Y)=\best_\IS(\cB,Y)=-\infty$.
		
		Suppose from now that $\best_\IS(\cA,Y)\neq -\infty$ and let $X\in \cA$ such that $X\cup Y$ is an independent set and $w(X)=\best_\IS(\cA,Y)$.
		We distinguish the following cases:
		\begin{itemize}
			\item If $\umim(A)=k$, then, by construction, there exists an independent set $W\in \cB$ such that $N(X)\setminus A = N(W) \setminus A$ and $w(X)\le w(W)$.
			As $X\cup Y$ is an independent set, we deduce that $N(X)\cap Y=N(W)\cap Y = \emptyset$ and thus $W\cup Y$ is an independent set.
			
			\item If $\umim(A)> k$, then $\umim(\comp{A})=k$ as $\omim(A)=\min(\umim(A),\umim(\comp{A}))=k$.
			By \Cref{lem:umim}, there exists an independent set $Y'\subseteq Y $ of size at most $k$ such that $N(Y)\setminus \comp{A} = N(Y')\setminus \comp{A}$.
			As $Y'\subseteq Y$, we know that $X\cup Y'$ is an independent set. Thus, by construction there exists a set $W\in \cB$ such that $W\cup Y'$ is an independent set and $w(X)\le w(W)$.
			Since $N(Y)\setminus A = N(Y')\setminus A$, we deduce that $W\cup Y$ is an independent set.
		\end{itemize}
		In both cases, there exists $W\in \cB$ such that $W\cup Y$ is an independent set and $w(X)\le w(W) \le \best_\IS(\cB,Y)$.
		Since $\best_\IS(\cB,Y)\le \best_\IS(\cA,Y) = w(X)$, it follows that $w(X)=\best_\IS(\cA,Y)=\best_\IS(\cB,Y)$.
		As this holds for every $Y\subseteq \comp{A}$, we conclude that $\cB$ $(\IS,A)$-represents $\cA$.
		
		\paragraph{Running time.}
		Computing $\omim(A)=k$ and checking whether $\umim(A) = k$ can be done by looking at every set of $k+1$ edges and check whether one of these sets is an induced $(A,\comp{A})$-matching in $G - E(\comp{A})$ and in $G-E(A)$.
		This can be done in time $\OO(\binom{n^{2}}{k+1}n^{2})=n^{\OO(k)}$ time.
		When $\umim(A)> k$, it is clear that computing $\cB$ can be done in time $\abs{\cA} n^{\OO(k)}$.
		This is also possible when $\umim(A)=k$ as \Cref{lem:umim} implies that $\abs{\{N(X)\setminus A \mid X$ is an independent set in $\cA\}}\le n^{k}$.
	\end{proof}
	
	We obtain the following by using \Cref{thm:reduce} with the routine from \Cref{lem:omim:IS}.
	
	\begin{theorem}\label{thm:omim:IS}
		Given an $n$-vertex graph with a branch decomposition of o-mim-width $k$, we can solve \textsc{Independent Set} in time $n^{\OO(k)}$.
	\end{theorem}
	
	\subsection{Feedback Vertex Set}
	
	\def\FVS{\mathsf{\cF}}
	
	As usual, instead of looking for a minimum feedback vertex set, we look for an induced forest (the complement of a feedback vertex set) of maximum weight.
	We denote by $\FVS$ the collection of all the sets $X\subseteq V(G)$ that induces a forest.
	
	We start by proving that the edges of an induced forest crossing a cut of o-mim-width $k$ can be covered by at most $4k$ vertices.
	In fact, we prove a stronger result by using sim-width which is always smaller than o-mim-width.
	To prove this property, we need the following notion of important vertices.
	\begin{definition}\label{def:fvs:Aimportant}
		Let $F$ be an induced forest of $G$. We call $x\in V(F)\cap A$ (1)~an $A$-internal vertex when $x$ has at least two neighbors in $V(F)\cap \comp{A}$, (2)~$A$-pendant when $x$ has only one neighbor $y$ in $V(F)\cap \comp{A}$ and $x$ is the only vertex from $V(F)\cap A$ adjacent to $y$ and  (3)~$A$-important when $x$ is $A$-internal or $A$-pendant.
	\end{definition}
	Observe that by definition, every edge of an induced forest $F$ between $A$ and $\comp{A}$ is incident to an $A$-important vertex or an $\comp{A}$-important vertex of $F$. 
	The following lemma provides an upper-bound on the number of $A$-important vertices of an induced forest.
	
	\begin{lemma}\label{lem:fvs:important}
		For every induced forest $F$ of $G$, the number of $A$-important vertices of $F$ is at most $2\sw(A)$.
	\end{lemma}
	\begin{proof}
		Let $X$ be the set of $A$-important vertices of an induced forest $F$ of $G$.
		We construct a bipartition $(X_0,X_1)$ of $X$ and we associate each vertex $x\in X$ with a vertex $y_x\in V(F)\cap \comp{A}$ such that the sets $\{x y_x \mid x\in X_0\}$ and $\{x y_x \mid x\in X_1\}$ are induced $(A,\comp{A})$-matchings of $G$.
		This is sufficient to prove that $\abs{X}\le 2 \sw(A)$.
		
		We construct $(X_0,X_1)$ by doing the following on each connected component $C$ of $F$.
		We fix a vertex $v_C\in V(C)\cap X$ that we add to $X_0$, then we do a breadth-first traversal of $C$ from $v_C$.
		When we visit a vertex $x\in X \cap V(C)$, we consider the path $P_x$ between $x$ and $v_C$ in $C$.
		If $x$ is $A$-pendant, we consider $y_x$ to be its unique neighbor in $V(F)\cap \comp{A}$.
		Otherwise $x$ is $A$-internal and it admits at least two neighbors in $V(F)\cap \comp{A}$, thus at least one of them---that we consider as $y_x$---does not lie in $P_x$.
		
		We define the parent of $x$ as follows:
		If (1)~$x$ is $A$-pendant, (2)~$y_x$ lies in $P_x$ and (3)~$y_x$ is adjacent to a vertex $y_w$ in $P_x$ associated with an $A$-pendant vertex $w$, we define the parent of $x$  as $w$. 
		Otherwise, we define the parent of $x$ as the $A$-important vertex that is the closest to $x$ in $P_x$.
		We add $x$ to the set among $(X_0,X_1)$ which does not contain its parent, this is well-defined because the parent of $x$ is always closer to $v_C$ than $x$ and thus $x$ is visited after its parent.
		
		Assume towards a contradiction that $\{x y_x \mid x\in X_0\}$ is not an induced $(A,\comp{A})$-matching of $G$ (the proof is symmetrical for $\{x y_x \mid x\in X_1\}$).
		Thus, there exist two distinct $A$-important vertices $x_1,x_2\in X_0$ such that $\{x_1y_{x_1},x_2y_{x_2}\}$ is not an induced matching.
		This implies that the vertices $x_1,x_2,y_{x_1},y_{x_2}$ induce a connected graph and they belong to the same connected component $C$ of $F$.
		Since $C$ is a tree, we deduce that there exists $i\in\{1,2\}$ such that $x_{3-i}$ or $y_{x_i}$ lies in the path $P_{x_i}$ between $x_i$ and $v_C$.
		Without loss of generality, assume that $x_1$ or $y_{x_2}$ lies in $P_{x_2}$.
		We prove that either $x_1$ is the parent of $x_2$ or $x_2$ is the parent of $x_1$, in both cases, this yields a contradiction since both $x_1$ and $x_2$ belong to $X_0$.
		\begin{itemize}
			\item Suppose first that $x_1$ lies in $P_{x_2}$. Since $x_1,x_2,y_{x_1},y_{x_2}$ induce a connected graph, we deduce that $x_1$ must be the parent of $x_2$, yielding a contradiction.
			
			\item Assume now that $y_{x_2}$ lies in $P_{x_2}$ and $x_1$ does not lie in $P_{x_2}$. 
			The choice of $y_{x_2}$ implies that $x_2$ is $A$-pendant: $y_{x_2}$ is the unique neighbor of $x_2$ in $V(F)\cap \comp{A}$ and $y_{x_2}$ is the unique neighbor of $y_{x_2}$ in $V(F)\cap A$.
			In particular, it implies that $x_1$ is not adjacent to $y_{x_2}$ and $x_2$ is not adjacent to $y_{x_1}$.
			As $\{x_1y_{x_1},x_2y_{x_2}\}$ is not an induced matching, either $x_1$ is adjacent to $x_2$ or $y_{x_1}$ is adjacent to $y_{x_2}$.
			If $x_1$ is adjacent to $x_2$, then $x_2$ is the neighbor of $x_1$ in $P_{x_1}$ and $x_2$ must be the parent of $x_1$, this yields a contradiction.
			
			Suppose that $y_{x_1}$ and $y_{x_2}$ are adjacent.
			As $x_1$ does not lie in $P_{x_2}$, $y_{x_1}$ must lie in $P_{x_1}$.
			Thus, $x_1$ is also a $A$-pendant vertex.
			It follows that either $y_{x_1}$ belongs in $P_{x_{2}}$ or $y_{x_2}$ belongs to $P_{x_1}$.
			Without loss of generality, assume that $y_{x_1}$ belongs in $P_{x_{2}}$.
			By definition of the parent of $x_2$, we deduce that $x_1$ is the parent of $x_2$ because the last three vertices of $P_{x_2}$ are $y_{x_1}, y_{x_2}$ and $x_2$, yielding a contradiction.
		\end{itemize}
		Hence, $\{x y_x \mid x\in X_0\}$ and by symmetry $\{x y_x \mid x\in X_1\}$ are induced $(A,\comp{A})$-matchings of $G$.
		We conclude that $\abs{X}\leq 2\sw(A)$.
	\end{proof}
	
	Our routine for computing small representative sets for \textsc{Feedback Vertex Set} is based on the following set of triples.
	\def\bX{\mathbb{X}}
	\def\bY{\mathbb{Y}}
	\def\bW{\mathbb{W}}

	\begin{definition}
		We define $\cT$ as the set of all triples $(\bX,\bY,\bW)$ such that:
		\begin{itemize}
			\item $\bX$ is a subset of $A$, $\bY$ is a subset of $\comp{A}$ and 
			
			\item If $\umim(A)=k$, then $\bW$ is a subset of $A$, otherwise $\bW$ is a subset of $\comp{A}$.
			
			\item The sizes of $\bX,\bY$ and $\bW$ are at most $2\omim(A)$.
		\end{itemize}
	\end{definition}
	The last item guarantees that $\cT$ contains at most $n^{6\omim(A)}$ triples.
	
	To compute small representative sets, we define a notion of compatibility between the triples in $\cT$ and the subsets of $A$ and $\comp{A}$ and we define an equivalence relation $\sim_t$ between the subsets of $A$ compatible with a triple $t$  such that:
	\begin{itemize}
		\item For every induced forest $F$, there exists a triple in $\cT$ compatible with $V(F)\cap A$ and $V(F)\cap \comp{A}$.
		\item For every $X,W\subseteq A$ and $Y\subseteq \comp{A}$ compatible with $t\in \cT$, if $X\sim_t W$, then $X\cup Y$ induces a forest iff $W\cup Y$ induces a forest.
	\end{itemize}
	Given $A\subseteq V(G)$ and $\cA\subseteq 2^{A}$, we compute a small representative set $\cB$ of $\cA$ from the empty set by adding, for each triple $t$ of $\cT$ and equivalence class $\cC$ of $\sim_t$, a set $X\in \cA\cap \cC$ of maximum weight.
	The above-mentioned properties guarantee that $\cB$ $(\FVS,A)$-represents $\cA$.
	
	Based on \Cref{lem:umim,lem:fvs:important}, our notion of compatibility guarantee that, for every $X\subseteq A$ and $Y\subseteq \comp{A}$ compatible with $t=(\bX,\bY,\bW)$,  the set $\bX\cup \bY$ is a vertex cover of $G[X,Y]$, i.e. every edge between $X$ and $Y$ has at least one endpoint in $\bX\cup \bY$.
	Moreover, two subsets of $A$ are equivalent for $\sim_t$ if they connect the vertices of $\bX\cup \bY$ in the same way.
	The number of equivalence classes of $\sim_t$ is at most $(4k)^{4k}$ with $k=\omim(A)$ since $\abs{\bX\cup \bY} \leq 4k$.
	As $\abs{\cT} \leq n^{6k}$, the size of the computed representative set $\cB$ is at most $n^{6k} (4k)^{4k}$.
	
	We start by defining our notion of compatibility.
	
	\begin{definition}
		We say that a set $X\subseteq A$ is compatible with $(\bX,\bY,\bW)\in \cT$ if:
		\begin{enumerate}
			\item\label{item:A:forest} $\bX\subseteq X$ and the graph $G[X\cup \bY]$ is a forest.
			\item\label{item:A:W} If $\bW \subseteq A$, then $N(X\setminus \bX) \cap \comp{A} = N(\bW)\cap \comp{A}$, otherwise, $N(\bW)\cap (X\setminus \bX) =\emptyset$.
		\end{enumerate}
		Moreover, we say that a set $Y\subseteq \comp{A}$ is compatible with $(\bX,\bY,\bW)\in \cT$ if:
		\begin{enumerate}[A.]
			\item\label{item:compA:forest} $\bY\subseteq Y$ and the graph $G[\bX\cup Y]$ is a forest.
			\item\label{item:compA:W} If $\bW \subseteq \comp{A}$, then $N(Y\setminus \bY) \cap A = N(\bW)\cap  A$, otherwise $N(\bW)\cap (Y\setminus \bY) =\emptyset$.
		\end{enumerate}
	\end{definition}
	
	The following lemma proves the most important property of our notion of compatibility.
	
	\begin{lemma}\label{lem:fvs:triple}
		Let $X\subseteq A$ and $Y\subseteq \comp{A}$. If $G[X\cup Y]$ induces a forest, then there exists a triple in $\cT$ compatible with $X$ and $Y$.
	\end{lemma}
	\begin{proof}
		Assume that $F=G[X\cup Y]$ is a forest.
		We construct a triple $(\bX,\bY,\bW)\in \cT$ as follows.
		We set $\bX$ as the set of all $A$-important vertices of $F$ and $\bY$ as the set of all  $\comp{A}$-important vertices of $F$.
		By \Cref{lem:fvs:important}, the sizes of $\bX$ and $\bY$ are at most $2k$.
		We define $\bW$ as follows:
		\begin{itemize}
			\item If $\umim(A)=k$, we consider $\bW$ as a subset of $X\setminus \bX$ of size at most $2k$ such that $N(\bW)\setminus A = N(X\setminus \bX)\setminus A$.
			\item Otherwise, $\umim(\comp{A})=k$ and we consider $\bW$ as a subset of $Y\setminus \bY$ of size at most $2k$ such that $N(\bW)\cap A = N(Y\setminus \bY)\setminus A$.
		\end{itemize}
		The existence of $\bW$ is guaranteed by \Cref{lem:umim} and the fact that $X\setminus \bX$ and $Y\setminus \bY$ can be partitioned into two independent sets since both sets induce forests.
		
		Since the sizes of $\bX,\bY$ and $\bW$ are at most $2k$, we have $(\bX,\bY,\bW)\in \cT$.
		Since $\bX\subseteq X$, $\bY\subseteq Y$ and $X\cup Y$ induced a forest, we deduce that Property~\ref{item:A:forest} and~\ref{item:compA:forest} are satisfied.
		It remains to prove Properties~\ref{item:A:W} and~\ref{item:compA:W}.
		Since these properties are symmetric, we assume without loss of generality that $\umim(A)=k$.
		Thus, we have $N(\bW)\cap \comp{A} = N(X\setminus \bX) \cap \comp{A}$ and Property~\ref{item:A:W} is satisfied.
		By definition, every edge between $X$ and $Y$ has at least one endpoint which is $A$-important or $\comp{A}$-important in $F$.
		It follows that $N(X\setminus \bX)\cap (Y\setminus \bY)=\emptyset$.
		We deduce that $N(\bW)\cap (Y\setminus \bY)=\emptyset$ and thus Property~\ref{item:compA:W} is satisfied.
		Hence, $X$ and $Y$ are compatible with $(\bX,\bY,\bW)\in \cT$.
	\end{proof}

	We associate every triple of $\cT$ with the equivalence relation defined below.
	
	\begin{definition}
		For every $X,W\subseteq V(G)$ compatible with a triple $t=(\bX,\bY,\bW)$, we say that $X$ and $W$ are $t$-equivalent if for every $u,v\in \bX\cup \bY$, $u$ and $v$ are connected in $G[X\cup \bY]$ iff $u$ and $v$ are connected in $G[W\cup \bY]$.
	\end{definition}
	
	The following lemma proves that two $t$-equivalent partial solutions give a forest with the same subsets of $\comp{A}$ compatible with $t$.
	
	\begin{lemma}\label{lem:fvs:equivalent}
		Let $t=(\bX,\bY,\bW)\in \cT$ and $X,W\subseteq A$ compatible with $t$. If $X$ and $W$ are $t$-equivalent, then, for every $Y\subseteq \comp{A}$ compatible with $t$, we have $G[X\cup Y]$ is a forest if and only if $G[W\cup Y]$ is a forest.
	\end{lemma}
	\begin{proof}
		Let $Y\subseteq \comp{A}$ compatible with $t$. 
		Assume that $X$ and $W$ are $t$-equivalent and $G[X\cup Y]$ contains a cycle $C$.
		According to Property~\ref{item:A:forest}, the graph $G[X\cup \bY]$ is a forest, hence $C$ contains at least one vertex $y$ in $Y\setminus \bY$.
		Properties~\ref{item:A:W} and~\ref{item:compA:W} guarantee that there is no edge between $X\setminus \bX$ and $Y\setminus \bY$.
		Consequently, every edge between $X$ and $Y$ has an endpoint in $\bX\cup \bY$.
		We deduce that $C$ is the concatenation of $\ell\ge 1$ edge-disjoint paths $P_1,\dots,P_\ell$ such that for each $i\in [\ell]$ we have:
		\begin{itemize}
			\item $P_i$ is a non-empty \textit{path} with endpoints in $\bX\cup \bY$ and internal vertex not in $\bX\cup \bY$ and $P_i$ is a path of $G[X\cup \bY]$ or $G[\bX \cup Y]$.
		\end{itemize}
		Without loss of generality, suppose that $P_1$ is the path going through $y$.
		Every path $P_i$ (including $P_1$) that lies in $G[\bX\cup Y]$ is a path of $G[W\cup Y]$ because $\bX\subseteq W$.
		Moreover, as $X$ and $W$ are $t$-equivalent, every path $P_i$ lying in $G[X\cup \bY]$ can be replaced by a path in $G[W\cup \bY]$.
		By applying these replacements on the concatenation of the paths $P_2,\dots,P_\ell$ we obtain a walk $P$ of $G[W\cup Y]$ between the endpoints of $P_1$ such that $P_1$ and $P$ are edge-disjoint.
		Hence, $G[W\cup Y]$ contains a cycle.
		We conclude that $G[X\cup Y]$ is a forest if and only if $G[W\cup Y]$ is a forest.
	\end{proof}
	
	The following lemma shows how to compute small $(\FVS,A)$-representative sets.
	
	\begin{lemma}\label{lem:fvs:reduce}
		Let $\omim(A)=k$.
		Given a collection $\cA\subseteq 2^{A}$, we can compute in time $\abs{\cA} n^{\OO(k)}$ a subset $\cB$ of $\cA$ such that $\cB$ $(\FVS,A)$-represents $\cA$ and $\abs{\cB}\le n^{6 k} (4k)^{4k}$.
	\end{lemma}
	\begin{proof}
		Let $\cA\subseteq 2^{A}$.
		We compute $\cB$ from the empty set as follows: for each triple $t\in \cT$ and $t$-equivalence class $\cC$ over $\cA$, we add to $\cB$ a set $X\in \cC$ of maximum weight.
		
		\paragraph*{Correctness.} 		Since every triple of $\cT$ consist of 3 sets of vertices whose sizes are at most $2k$, we have $\abs{\cT}\leq n^{6k}$.
		For each triple $(\bX,\bY,\bW)\in \cT$, the number of $t$-equivalence classes is the number of partitions of $\bX\cup \bY$ which is at most $(4k)^{4k}$.
		We deduce that $\abs{\cB}\leq n^{6 k} (4k)^{4k}$.
		
		To prove that $\cB$ $(\FVS,A)$-represents $\cA$, we need to prove that for every $Y\subseteq \comp{A}$, we have $\best_\FVS(\cA,Y)=\best_\FVS(\cB,Y)$.
		Let $Y\subseteq \comp{A}$.
		Assume that $\best_\FVS(\cA,Y)\neq \emptyset$ (otherwise $\best_\FVS(\cA,Y)=\best_\FVS(\cB,Y)=-\infty$ because $\cB\subseteq \cA$).
		Let $X\in\cA$ such that $G[X\cup Y]$ is a forest and $w(X)=\best_{\FVS}(\cA,Y)$.
		\Cref{lem:fvs:triple} implies the existence of a triple $t\in \cT$ compatible with $X$ and $Y$.
		By construction, $\cB$ contains a set $W$ compatible with $t$ such that $w(X)\leq w(W)$ and $X$ and $W$ are $t$ equivalent.
		From \Cref{lem:fvs:equivalent}, we deduce that $G[W\cup Y]$ is a forest.
		We conclude that $w(X)=\best_{\FVS}(\cA,Y)=\best_{\FVS}(\cB,Y)=w(W)$.	
		
		\paragraph*{Running time.} We can enumerate $\cT$ in time $n^{\OO(k)}$.
		Moreover, for each triple $t\in \cT$, checking whether two partial solutions are $t$-equivalent can be done in time $\OO(n^{2})$.
		As $\abs{\cB}\leq n^{6k}(4k)^{4k}$, we deduce that $\cB$ can be computed in time $\abs{\cA}n^{\OO(k)}$ with standard algorithmic techniques.
	\end{proof}
	
	The next theorem follows from  \Cref{thm:reduce,lem:fvs:reduce}.
	\begin{theorem}\label{thm:omim:fvs}
		Given an $n$-vertex graph with a branch decomposition of o-mim-width $w$, we can solve \textsc{Feedback Vertex Set} in time $n^{\OO(w)}$.
	\end{theorem}
	Note that \Cref{thm:omim:algo} is the combination of \cref{thm:omim:IS,thm:omim:fvs}.

	\subsection{Relation between o-mim-width and tree-independence number}
	\label{subsec:relomimwidth}
	We show that the o-mim-width of a graph is upper bounded by its tree-independence number.
	We start with the standard notion of nice tree decompositions.
	
	A rooted tree decomposition is a tree decomposition where one node $r \in V(T)$ is designated as the root.
	A nice tree decomposition is a rooted tree decomposition, where the bag of the root node and every leaf node is empty, and every other node is either (1)~an \emph{introduce-node} that has one child node and whose bag is the bag of the child node plus one vertex, (2)~a \emph{forget-node} that has one child node and whose bag is the bag of the child node minus one vertex, or (3)~a \emph{join-node} that has two children and whose bag is equal to the bags of both of its children.
	The following classic lemma shows that any tree decomposition can be turned into a nice tree decomposition.
	
	\begin{lemma}[See e.g.~\cite{DBLP:books/sp/CyganFKLMPPS15}]
		\label{lem:nicetd}
		Let $(T, \bag)$ be a tree decomposition of a graph $G$.
		There exists a nice tree decomposition $(T', \bag')$ of $G$ so that every bag of $(T', \bag')$ is a subset of some bag of $(T,\bag)$.
	\end{lemma}
	
	We say that a branch decomposition is \emph{on a set} $V(G)$ if it is a branch decomposition of some function $\f : 2^{V(G)} \rightarrow \bZ_{\ge 0}$.
	Next we give a general lemma for turning tree decompositions of $G$ into branch decompositions on $V(G)$.
	
	\begin{lemma}
		\label{lem:tdintobd}
		Let $(T, \bag)$ be a tree decomposition of a graph $G$.
		There exists a branch decomposition $(T', \delta)$ on the set $V(G)$ so that  for every bipartition $(A,\overline{A})$ of $V(G)$ given by an edge of $(T', \delta)$, there exists a bag of $(T, \bag)$ that contains either $N(A)$ or $N(\overline{A})$. 
	\end{lemma}
	\begin{proof}
		First, we can without loss of generality use Lemma~\ref{lem:nicetd} to assume that $(T, \bag)$ is a nice tree decomposition.
		Now, observe that because of the condition (2) of tree decompositions, there is a bijection between forget-nodes and the vertices $V(G)$, i.e., for every $v \in V(G)$ there exists exactly one forget-node $f_v \in V(T)$ with a child $c_v$  so that $\bag(f_v) = \bag(c_v) \setminus \{v\}$.
		Now, we construct a tree $T''$ from $T$ by inserting a leaf node $l_v$ adjacent to each $f_v$ and then construct $\delta'$ by mapping $v$ to $l_v$.
		The pair $(T'', \delta')$ is not yet a branch decomposition because it contains leaves to which no vertices are mapped and internal degree-2 vertices.

		Before turning $(T'', \delta')$ into a branch decomposition, let us first prove that if $(A,\overline{A})$ is a bipartition of $V(G)$ corresponding to an edge of $T''$, then there exists a bag of $(T, \bag)$ that contains either $N(A)$ or $N(\overline{A})$.
		First, if the edge is between $f_v$ and $l_v$ for some vertex $v \in V(G)$, then the bipartition is of form $(A,\overline{A}) = (\{v\}, V(G) \setminus \{v\})$, and therefore $N(\overline{A}) \subseteq \{v\}$ is contained in a bag of $(T,\bag)$.
		Otherwise, the edge corresponds to an edge of $T$ between some node $t$ and its parent $p$, and we can orient $(A,\overline{A})$ so that $A$ contains the vertices $v \in V(G)$ so that $f_v$ is in the subtree rooted at $t$.
		In this case, we show that $N(A) \subseteq \bag(t)$.
		Let $uv \in E(G)$ so that $u \in A$ and $v \in \overline{A}$.
		Now, the vertex $u$ occurs only in the bags of a subtree rooted at some child of $t$.
		However, $v$ occurs also somewhere else (because $v \in \overline{A}$), so to satisfy the conditions of tree decompositions, $v$ must occur in $\bag(t)$.

		Then, we can turn $(T'', \delta')$ into a branch decomposition $(T', \delta)$ by iteratively deleting leafs to which no vertices are mapped and suppressing degree-2 vertices.
		This does not affect the set of bipartitions of $V(G)$ corresponding to the edges of the decomposition.
	\end{proof}
	
	Then we restate Theorem~\ref{thm:omimtin} and prove it using Lemma~\ref{lem:tdintobd}.
	
	\omimtin*
	\begin{proof}
		Let $G$ be a graph with tree-independence number $k$ and $(T, \bag)$ a tree decomposition of $G$ with independence number $\alpha(T, \bag) = k$.
		By applying Lemma~\ref{lem:tdintobd} we turn $(T, \bag)$ into a branch decomposition on $V(G)$ so that for every partition $(A,\overline{A})$ of $V(G)$ given by the decomposition, either $N(A)$ or $N(\overline{A})$ has independence number at most $k$.
		Now, if $N(A)$ has independence number at most $k$, then $\umim(\overline{A}) \le k$, and if $N(\overline{A})$ has independence number at most $k$, then $\umim(A) \le k$, so we have that $\omim(A) \le k$, and therefore the o-mim-width of the branch decomposition is at most $k$.
	\end{proof}

	We also prove the following.
	
	\begin{theorem}\label{thm:simwVsMMHTW}
		Any graph with minor-matching hypertreewidth $k$ has sim-width at most $k$.
	\end{theorem}
	\begin{proof}
		Let $G$ be a graph with minor-matching hypertreewidth $k$ and $(T, \bag)$ a tree decomposition of $G$ with $\mu(T,\bag) = k$.
		By applying Lemma~\ref{lem:tdintobd} we turn $(T, \bag)$ into a branch decomposition on $V(G)$ so that for every partition $(A,\overline{A})$ of $V(G)$ given by the decomposition, either $N(A)$ or $N(\overline{A})$ is contained in a bag of $T$, and therefore either $\mu(N(A)) \le k$ or $\mu(N(\overline{A})) \le k$.
		Note that however, if there would be an induced matching of size $k+1$ between $A$ and $\overline{A}$, then both $\mu(N(A)) > k$ and $\mu(N(\overline{A})) > k$.
		Therefore, the branch decomposition has sim-width at most $k$.
	\end{proof}
	
	\section{Neighbor-depth}
	We start this section with the definition of the graph parameter neighbor-depth.
	
	\begin{definition}
		\label{def:nd}
		The neighbor-depth (\nd) of a graph $G$ is defined recursively as follows:
		\begin{enumerate}
			\item $\nd(G) = 0$ if and only if $V(G) = \emptyset$,
			\item \label{def:nd:case2} if $G$ is not connected, then $\nd(G)$ is the maximum value of $\nd(G[C])$ where $C \subseteq V(G)$ is a connected component of $G$,
			\item \label{def:nd:case3} if $V(G)$ is non-empty and $G$ is connected, then $\nd(G) \le k$ if and only if there exists a vertex $v \in V(G)$ such that $\nd(G \setminus N[v]) \le k-1$ and $\nd(G \setminus \{v\}) \le k$.
		\end{enumerate}
	\end{definition}
	
	In the case (\ref{def:nd:case3}) of Definition~\ref{def:nd}, we call the vertex $v$ the \emph{pivot-vertex} witnessing $\nd(G) \le k$.

	\subsection{Algorithms using neighbor-depth}
	We start by showing that neighbor-depth is monotone under induced subgraphs.
	
	\begin{lemma}
		\label{lem:indsub}
		Let $G$ be a graph and $X \subseteq V(G)$. It holds that $\nd(G[X]) \le \nd(G)$.
	\end{lemma}
	\begin{proof}
		We use induction on $|V(G)|$.
		The lemma clearly holds when $X = \emptyset$, so the base case of $|V(G)| = 0$ holds and we can assume that $|X| \ge 1$.
		We can assume that $G[X]$ is connected by taking the connected component $X' \subseteq X$ of $G[X]$ with $\nd(G[X']) = \nd(G[X])$.
		
		First, if $G$ is not connected, let $C \subseteq V(G)$ be the connected component of $G$ such that $X \subseteq C$.
		As $|C| < |V(G)|$, we have $\nd(G[X]) \le \nd(G[C])$ by the induction assumption.
		Since $C$ is a connected component of $G$, we have by definition $\nd(G[C]) \le \nd(G)$, and therefore $\nd(G[X]) \le \nd(G)$.

		Second, if $G$ is connected let $v$ be the pivot-vertex witnessing the neighbor-depth of $G$.
		If $v \notin X$, then $X \subseteq V(G) \setminus \{v\}$ and the lemma holds by induction.
		Otherwise, if $v \in X$, then $X \setminus N[v] \subseteq V(G) \setminus N[v]$ and $X \setminus \{v\} \subseteq V(G) \setminus \{v\}$, so the lemma holds by choosing $v$ as the pivot-vertex of $G[X]$ and induction.
	\end{proof}
	
	Then, we show that the neighbor-depth of a graph $G$ can be computed in $n^{\OO(\nd(G))}$ time.
	
	\begin{lemma}
		\label{lem:ndcomp}
		Given an $n$-vertex graph $G$ and integer $k$, it can be decided whether $\nd(G) \le k$ in $\OO(n^{2k+3})$ time.
		In the case when $\nd(G) \le k$ and $G$ is non-empty and connected, the algorithm also outputs the pivot-vertex witnessing the neighbor-depth.
	\end{lemma}
	\begin{proof}
		The lemma trivially holds for $k = 0$.
		When $k > 0$, we prove the lemma by induction, using the algorithm for checking $\nd(G) \le k-1$ as a subroutine.
		If $G$ is not connected, we solve each connected component independently.
		If $G$ is connected, we use the algorithm for checking $\nd(G) \le k-1$ to test for each $v \in V(G)$ whether $\nd(G \setminus N[v]) \le k-1$.
		If no such $v$ is found, we have by definition $\nd(G) > k$, and therefore return false.
		If such $v$ is found, we have that $\nd(G) \le k$ if and only if $\nd(G \setminus \{v\}) \le k$, where the if direction is by definition and only if direction is by Lemma~\ref{lem:indsub}.
		Therefore we remove $v$ from the graph, and continue the process in each connected component of $G \setminus \{v\}$.
		We make at most $n^2$ calls to the subroutine for checking $\nd(G) \le k-1$, so by induction the total size of the recursion tree is at most $n^{2k}$, and therefore as each recursive step can be implemented in $\OO(n^3)$ time, the total time complexity is $\OO(n^{2k+3})$
	\end{proof}

	\begin{lemma}
		\label{lem:isneighdepth}
		\textsc{Independent set} can be solved in time $n^{\OO(k)}$ on $n$-vertex graphs with neighbor-depth at most $k$.
	\end{lemma}
	\begin{proof}
		We describe a recursive algorithm that given a graph $G$ returns an optimal independent set in $G$.
		Clearly, if $G$ is empty we can return the empty set, and if $G$ is not connected we can return the union of an optimal independent set in each connected component.
		When $G$ is connected and non-empty, we use \Cref{lem:ndcomp} to compute a pivot-vertex $v$ witnessing $\nd(G)\le k$.
		Then, we compute an optimal independent set $I_v$ of $G\setminus N[v]$ and an optimal independent set $I_{\neg v}$ of $G\setminus \{v\}$.
		If $w(I_v)+w(v) \ge w(I_{\neg v})$, we output $I_v \cup \{v\}$ and otherwise, we output $I_{\neg v}$.
		
		The correctness of this algorithm follows from the fact that for every $v\in V(G)$, an optimal independent set is either (1)~the union of $\{v\}$ and an optimal independent set of $G\setminus N[v]$ or (2)~an optimal independent set of $G\setminus \{v\}$.
		
		As $\nd(G\setminus N[v])\leq k-1$ for every pivot-vertex $v$ witnessing $\nd(G)\le k$, we deduce that we generate at most $n^2$ recursive calls on graphs of neighbor-depth $k-1$.
		Hence, the total size of the recursion tree is at most $n^{2k}$.
		By \Cref{lem:ndcomp}, computing a pivot-vertex can be done in time $n^{2k+3}$. Hence, each recursive step can be implemented in  $n^{O(k)}$ time. 
		We conclude that the running time of this algorithm is $n^{O(k)}$.
	\end{proof}

	\subsection{Neighbor-depth of graphs of bounded sim-width}
	\label{subsec:ndsimwidth}
	In this subsection we show that graphs of bounded sim-width have poly-logarithmic neighbor-depth, i.e., Theorem~\ref{thm:simwidthnbdepth}.
	The idea of the proof will be that given a cut of bounded sim-width, we can delete a constant fraction of the edges going over the cut by deleting the closed neighborhood of a single vertex.
	This allows to first fix a balanced cut according to an optimal decomposition for sim-width, and then delete the edges going over the cut in logarithmic depth.

	We say that a vertex $v \in V(G)$ \emph{neighbor-controls} an edge $e \in E(G)$ if $e$ is incident to a vertex in $N[v]$.
	In other words, $v$ neighbor-controls $e$ if $e \notin E(G \setminus N[v])$.
	
	\begin{lemma}
		\label{lem:control}
		Let $G$ be a graph and $A \subseteq V(G)$ so that $\sw(A) \le k$.
		There exists a vertex $v \in V(G)$ that neighbor-controls at least $|E(A,\overline{A})|/2k$ edges in $E(A,\overline{A})$. 
	\end{lemma}
	\begin{proof}
		Suppose the contradiction, i.e., that all vertices of $G$ neighbor-control less than $|E(A,\overline{A})|/2k$ edges in $E(A,\overline{A})$.
		Let $M \subseteq E(A,\overline{A})$  be a maximum induced $(A,\overline{A})$-matching, having size at most $|M| \le \sw(A) \le k$, and let $V(M)$ denote the set of vertices incident to $M$.
		Now, an edge in $E(A,\overline{A})$ cannot be added to $M$ if and only if one of its endpoints is in $N[V(M)]$.
		In particular, an edge in $E(A,\overline{A})$ cannot be added to $M$ if and only if there is a vertex in $V(M)$ that neighbor-controls it.
		However, by our assumption, the vertices in $V(M)$ neighbor-control strictly less than 
		\[|V(M)| \cdot |E(A,\overline{A})|/2k = |E(A,\overline{A})|\]
		edges of $E(A,\overline{A})$, so there exists an edge in $E(A,\overline{A})$ that is not neighbor-controlled by $V(M)$, and therefore we contradict the maximality of $M$.
	\end{proof}

	Now, the idea will be to argue that because sim-width is at most $k$, there exists a balanced cut $(A,\overline{A})$ with $\sw(A) \le k$, and then select the vertex $v$ given by Lemma~\ref{lem:control} as the pivot-vertex.
	Here, we need to be careful to persistently target the same cut until the graph is disconnected along it.
	
	\simwidthnbdepth*
	\begin{proof}
		For integers $n \ge 2$ and $k, t \ge 0$, we denote by $\nd(n, k, t)$ the maximum neighbor-depth of a graph that
		\begin{enumerate}
			\item\label{enum:swnbproof:vert} has at most $n$ vertices,
			\item has sim-width at most $k$, and
			\item\label{enum:swnbproof:sep} has a cut $(A,\overline{A})$ with $\sw(A) \le k$, $|E(A,\overline{A})| \le t$, $|A| \le 2n/3$, and $|\overline{A}| \le 2n/3$.
		\end{enumerate}
		
		We observe that if a graph $G$ satisfies all of the conditions~\ref{enum:swnbproof:vert}-\ref{enum:swnbproof:sep}, then any induced subgraph of $G$ also satisfies the conditions.
		In particular, note that $n$ can be larger than $|V(G)|$, and in the condition~\ref{enum:swnbproof:sep}, the cut should be balanced with respect to $n$ but not necessarily with respect to $|V(G)|$.
		
		We will prove by induction that 
		\begin{equation}
			\label{eq:swnbproof:ind}
			\nd(n, k, t) \le 1+ 4k (\log_{3/2}(n) \cdot \log(n^2+1) + \log(t+1)).
		\end{equation}
		This will then prove the statement, because by Lemma~\ref{lem:balsep} any graph with $n$ vertices and sim-width $k$ satisfies the conditions with $t = n^2$.

		First, when $n \le 2$ this holds because any graph with at most two vertices has neighbor-depth at most one.
		We then assume that $n \ge 3$ and that Equation~(\ref{eq:swnbproof:ind}) holds for smaller values of $n$ and first consider the case $t = 0$.

		Let $G$ be a graph that satisfies the conditions~\ref{enum:swnbproof:vert}-\ref{enum:swnbproof:sep} with $t=0$.
		Because $t=0$, each connected component of $G$ has at most $2n/3$ vertices, and therefore satisfies the conditions with $n' = 2n/3$, $k' = k$, and $t' = (2n/3)^2$.
		Therefore, by induction each component of $G$ has neighbor-depth at most $\nd(2n/3, k, (2n/3)^2)$.
		Because the neighbor-depth of $G$ is the maximum neighbor-depth over its components, we get that
		\begin{align*}
			\nd(G) \le& \nd(2n/3, k, (2n/3)^2)\\
			\le& 1+ 4k (\log_{3/2}(2n/3) \cdot \log((2n/3)^2+1) + \log((2n/3)^2+1))\\
			\le& 1+ 4k ((\log_{3/2}(n)-1) \cdot \log((2n/3)^2+1) + \log((2n/3)^2+1))\\
			\le& 1 + 4k (\log_{3/2}(n) \cdot \log((2n/3)^2+1))\\
			\le& 1 + 4k (\log_{3/2}(n) \cdot \log(n^2+1)),
		\end{align*}
		which proves that Equation~(\ref{eq:swnbproof:ind}) holds when $t=0$.

		We then consider the case when $t \ge 1$.
		Assume that Equation~(\ref{eq:swnbproof:ind}) does not hold and let $G$ be a counterexample that is minimal under induced subgraphs.
		Note that this implies that $G$ is connected, and every proper 
		induced subgraph $G'$ of $G$ has neighbor-depth at most $1+ 4k (\log_{3/2}(n) \cdot \log(n^2+1) + \log(t+1))$.
		We can also assume that $t = |E(A,\overline{A})|$.

		Now, by Lemma~\ref{lem:control} there exists a vertex $v \in V(G)$ that neighbor-controls at least $t/2k$ edges in $E(A, \overline{A})$.
		We will select $v$ as the pivot-vertex.
		By the minimality of $G$, we have that $\nd(G \setminus \{v\}) \le 1+ 4k (\log_{3/2}(n) \cdot \log(n^2+1) + \log(t+1))$, so it suffices to prove that $\nd(G \setminus N[v]) \le 1+ 4k (\log_{3/2}(n) \cdot \log(n^2+1) + \log(t+1)) - 1$.
		Because $v$ neighbor-controls at least $t/2k$ edges in $E(A, \overline{A})$, the graph $G \setminus N[v]$ satisfies the conditions with $n' = n$, $k' = k$, and $t' = t - t/2k$.
		We denote 
		\[\alpha = \frac{t'+1}{t+1} = 1 - \frac{t/2k}{t+1} \le 1 - \frac{t/2k}{2t} \le 1 - \frac{1}{4k}.\]
		Now we have that
		\begin{align*}
			\nd(G) \le& \nd(n, k, t-t/2k) + 1\\
			\le& 2 + 4k (\log_{3/2}(n) \cdot \log(n^2+1) + \log(\alpha \cdot (t+1)))\\
			\le& 2 + 4k (\log_{3/2}(n) \cdot \log(n^2+1) + \log(\alpha) + \log(t+1))\\
			\le& 2 + 4k \log(\alpha) + 4k (\log_{3/2}(n) \cdot \log(n^2+1) + \log(t+1))\\
			\le& 2 - 4k \cdot \frac{1}{4k} + 4k (\log_{3/2}(n) \cdot \log(n^2+1) + \log(t+1))\\
			\le& 1 + 4k (\log_{3/2}(n) \cdot \log(n^2+1) + \log(t+1)),
		\end{align*}
		which proves that Equation~(\ref{eq:swnbproof:ind}) holds when $t \ge 1$, and therefore completes the proof.
	\end{proof}

	\subsection{Neighbor-depth and independent set branching trees}
	
	We define an \emph{independent set branching tree} on a graph $G$ to be a rooted binary tree where 
	
	\begin{enumerate}
		\item each node is labeled with an induced subgraph of $G$,
		\item the root is labeled with $G$,
		\item each leaf is labeled with the empty graph,
		\item if a non-leaf node is labeled with the induced subgraph $G[X]$, then either
		\begin{enumerate}
			\item the node is a \emph{branching node}, in which case there exists a vertex $v \in X$ and the node has two children, with left child labeled with $G[X \setminus N[v]]$ and the right child labeled with $G[X \setminus \{v\}]$, or
			\item the node is a \emph{decomposition node}, in which case there exists a partition $(C_1, C_2)$ of $V(G)$ into two non-empty parts with no edges between and the node has two children labeled with $G[C_1]$ and $G[C_2]$.
		\end{enumerate}
	\end{enumerate}
	
	The size of an independent set branching tree is the number of nodes of it.
	Let $\beta(G)$ denote the smallest size of an independent set branching tree on $G$.
	We show that the neighbor-depth of $G$ both upper and lower bounds $\beta(G)$, in particular that $2^{\nd(G)} \le \beta(G) \le n^{\OO(\nd(G))}$.
	We start with the upper bound.
	
	\begin{lemma}
		There is an algorithm that given an $n$-vertex graph $G$, computes an independent set branching tree on $G$ of size $n^{\OO(\nd(G))}$ in time $n^{\OO(\nd(G))}$.
	\end{lemma}
	\begin{proof}
		Follows from observing that the algorithm of the proof of Lemma~\ref{lem:isneighdepth} constructs an independent set branching tree of size $n^{\OO(k)}$.
	\end{proof}

	We then prove the lower bound.
	
	\begin{lemma}
		\label{lem:branchinglb}
		Any independent set branching tree on a graph $G$ has at least $2^{\nd(G)}$ nodes.
	\end{lemma}
	\begin{proof}
		We prove the lemma by induction on the size of $G$.
		The base case of the empty graph holds, because the empty graph has neighbor-depth zero and branching tree of size one.
		Then, consider a branching tree of size $\beta(G)$ of a graph $G$.
		First, if the root node is a decomposition node with partition $(C_1, C_2)$, then one of $G[C_1]$ and $G[C_2]$ has neighbor-depth equal to $\nd(G)$, and the lower bound follows by induction.
		Then, if the root node is a branching node with branching vertex $v$, we have that $\nd(G \setminus N[v]) \ge \nd(G)-1$, and therefore also that $\nd(G \setminus \{v\}) \ge \nd(G)-1$, and therefore by induction both the subtree rooted at the left child and the subtree rooted at the right child have sizes at least $2^{\nd(G)-1}$, and therefore the branching tree of $G$ has size at least $2^{\nd(G)}$.
	\end{proof}
	
	
	
	\section{Separations between graph classes}
	\label{sec:paramrel}
	First, we show that $P_6$-free graphs can have unbounded sim-width.
	The same construction also excludes induced cycles of length 5 or more and induced complements of cycles of length 5 or more, i.e., is weakly chordal.
	This answers a question of Kang, Kwon, Str{\o}mme, and Telle about the sim-width of weakly chordal graphs~\cite[Question~3]{DBLP:journals/tcs/KangKST17}.
	
	\begin{proposition}
		\label{thm:swp6}
		For each $n$, there is a graph with $\OO(n^2)$ vertices and sim-width at least $\Omega(n)$ that does not contain induced paths of length 6, induced cycles of length 5 or more, or induced complements of cycles of length 5 or more.
	\end{proposition}
	\begin{proof}
		First, we take a complete bipartite graph with a bipartition $(V_1, V_2)$, with both $V_1$ and $V_2$ containing $n$ vertices.
		Then, for each pair $x,y$ with $x \in V_1$ and $y \in V_2$, we create a new degree-2 vertex $xy$ that is adjacent to $x$ and $y$.
		We claim that this constructions satisfies the statement.

		First, to prove that sim-width is at least $\Omega(n)$, let $k$ denote the sim-width of the construction and let us apply Lemma~\ref{lem:balsep} to find a cut $(A,\overline{A})$ that cuts the set $X = V_1 \cup V_2$ in a balanced manner, i.e., $|X \cap A| \le 4n/3$ and $|X \cap \overline{A}| \le 4n/3$, and has $\sw(A) \le k$.
		Let us permute $(A, \overline{A})$ so that $A$ is the side that contains the most degree-2 vertices, i.e., $A$ contains at least $n^2/2$ degree-2 vertices.
		Note that at most $(2/3)^2 n^2 = (4/9)
		n^2$ degree-2 vertices have both of their neighbors in $A \cap X$, so at least $n^2/2 - (4/9) n^2 = n^2/18$ of the degree-2 vertices in $A$ have at least one neighbor in $\overline{A} \cap X$.
		Therefore, for some $i \in \{1,2\}$, at least $n^2/36$ of the degree-2 vertices in $A$ have a neighbor in $\overline{A} \cap V_i$.
		Every vertex in $V_i$ has $n$ degree-2 neighbors, so there are at least $n/36$ different vertices in $\overline{A} \cap V_i$ that are adjacent to a degree-2 vertex in $A$.
		This gives an induced $(A,\overline{A})$-matching of size at least $n/36$, and therefore the sim-width is at least $n/36$.

		The constructed graph does not contain a $P_6$ because only its endpoints could be among the degree-2 vertices and a complete bipartite graph does not contain a $P_4$.
		Also, the constructed graph does not contain an induced cycle with more than four vertices because none of the degree-2 vertices can belong to such induced cycle, and a complete bipartite graph does not contain induced cycles longer than four.
		For complements of induced cycles, first recall that $\overline{C_5} = C_5$, so by previous sentence it does not contain complements of $C_5$.
		For $t \ge 6$, all vertices of $\overline{C_t}$ have degree more than two, so induced $\overline{C_t}$ could only use vertices of the complete bipartite graph, but it does not contain $\overline{C_t}$.
	\end{proof}
	
	It remains open whether $P_5$-free graphs have bounded sim-width.

	We then show that the minor-matching hypertreewidth of a graph class with bounded clique-width can be unbounded.

	\begin{proposition}
		For each $n$, there is a graph with $\OO(n)$ vertices, clique-width $\OO(1)$, and minor-matching hypertreewidth at least $n$.
	\end{proposition}
	\begin{proof}
		We take a complete bipartite graph $K_{n,n}$, and add a degree-1 pendant vertex adjacent to each vertex.
		The construction has bounded clique-width by e.g. first constructing two matchings,  one with labels 1 and 2 on different sides and one with labels 1 and 3 on different sides, and then taking a disjoint union of them and joining on labels 2 and 3.
		The construction has minor-matching hypertreewidth at least $n$, because there must be a bag that contains one side of the $K_{n,n}$ completely, but this gives an induced matching with at least $n$ edges intersecting the bag.
	\end{proof}
	
	We then show that the o-mim-width of a graph class with bounded minor-matching hypertreewidth can be unbounded, notice that by \Cref{thm:simwVsMMHTW}, such class have bounded sim-width.
	
	\begin{proposition}
		For each $n$, there is a graph with $\OO(n^5)$ vertices, minor-matching hypertreewidth $\OO(1)$, and o-mim-width at least $\Omega(n)$.
	\end{proposition}
	\begin{proof}
		Let $n \ge 1$ be an integer.
		First, we let $G'$ be the $n^4$-vertex graph that is a disjoint union of $n^2$ cliques each of size $n^2$.
		Then, we construct a graph $G$ as a path-like construction with $n$ levels, so that each even level induces a copy of $G'$, each odd level induces a copy of the complement of $G'$, and consecutive levels are connected by matchings that match the corresponding vertices together.

		We first show that $G$ has $\tmm$ at most $7$.
		Let $V_i$ denote the vertices of the $i$:th level, and let $i$ be odd.
		In particular, $G[V_i]$ is a copy of the complement of $G'$.
		First, we observe that the maximum induced matching in the graph $G[V_i]$ is of size one.
		Then, because the matching between two consecutive levels $V_i$ and $V_{i+1}$ maps vertices of $G[V_i]$ to corresponding vertices of its complement $G[V_{i+1}]$, the maximum induced $(V_i, V_{i+1})$-matching is of size one.
		It follows that $\mu(V_i) \le 3$, and therefore $\mu(V_i \cup V_{i+2}) \le 6$.
		
		We construct a tree decomposition of $G$ as follows.
		We start by creating a path decomposition where each bag contains the union of two consecutive odd levels, i.e., the vertices $V_i \cup V_{i+2}$, for odd $i$.
		The rest of the graph, i.e., the even levels, consists of connected components that are cliques of size $n^2$ whose neighborhoods are contained inside two consecutive odd levels.
		For each such clique $W$, we insert to the decomposition (making it now a tree decomposition instead of path decomposition) a new bag containing the clique and the odd levels adjacent to the clique, i.e., $B = W \cup V_i \cup V_{i+2}$ for some odd $i$, and make it adjacent to the bag containing $V_i \cup V_{i+2}$.
		Because $W$ is a clique, this bag has $\mu(B) \le \mu(W) + \mu(V_i) + \mu(V_{i+2}) \le 1 + 3 + 3 \le 7$, and therefore the constructed tree decomposition of $G$ has $\tmm$ at most $7$.

		We show that the o-mim-width of $G$ is at least $n/6$.
		Let $k$ be the o-mim-width of $G$.
		By Lemma~\ref{lem:balsep}, there exists a bipartition $(A, \overline{A})$ of $V(G)$ so that $\omim(A) \le k$, $|A| \ge n^5/3$, and $|\overline{A}| \ge n^5/3$.
		Because each odd level induces a connected induced subgraph and the odd levels are not adjacent to each other, we get that if all odd levels would intersect both $A$ and $\overline{A}$, there would be an induced $(A,\overline{A})$-matching of size $n$, and therefore the o-mim-width of $G$ would be at least $n$.
		In the other case, there is an odd level that is entirely contained in either $A$ or $\overline{A}$.
		By symmetry, suppose that there is an odd level entirely contained in $A$.
		Now, because $|\overline{A}| \ge n^5/3$, there must be some other level (even or odd) of which at least one third is contained in $\overline{A}$.
		By walking between these two levels, we find two consecutive levels so that at least a $1/(3n)$ fraction of the matching between them crosses the cut $(A, \overline{A})$.
		In particular, (after permuting $(A, \overline{A})$ if necessary) we find some $i$ so that there is a $(V_i \cap A, V_{i+1} \cap \overline{A})$-matching of size at least $n^3/6$.
		Let $M \subseteq V(G')$ be the vertices of $G'$ corresponding to this matching.
		We observe by the pigeonhole principle that $G'[M]$ contains both a clique of size at least $n/6$ and an independent set of size at least $n/6$.
		In particular, the endpoints of the matching must contain an independent set of size at least $n/6$ in both $V_i$ and $V_{i+1}$, which implies that $\umim(A) \ge n/6$ and $\umim(\overline{A}) \ge n/6$, implying that $\omim(A) \ge n/6$, implying that the o-mim-width of $G$ is at least $n/6$.
	\end{proof}

	We then give two results showing that even very limited graph classes have superconstant neighbor-depth.

	\begin{proposition}
		\label{pro:nbdpath}
		The $n$-vertex path has neighbor-depth at least $\Omega(\log n)$.
	\end{proposition}
	\begin{proof}
		We prove by induction that for all $k \ge 1$, the $3^k$-vertex path $P_{3^k}$ has neighbor-depth at least $k$.
		First, observe that this holds for $k=1$.
		Then, for $k \ge 2$, note that for any choice of $v \in V(P_{3^k})$, the graph $P_{3^k} \setminus N[v]$ contains the graph $P_{3^{k-1}}$ as an induced subgraph, and therefore by Lemma~\ref{lem:indsub} and induction has neighbor-depth at least $k-1$.
		Therefore, $P_{3^k}$ has neighbor-depth at least $k$, which again by Lemma~\ref{lem:indsub} implies that any $n$-vertex path has neighbor-depth at least $\Omega(\log n)$.
	\end{proof}
	
	We then show that cographs can have neighbor-depth $\Omega(\log n)$, showing together with Proposition~\ref{pro:nbdpath} that no graph classes in Figure~\ref{fig:classes} have constant neighbor-depth.
	Note that cographs have neighbor-depth at most $\OO(\log n)$, because any connected cograph contains a vertex of degree at least $n/2$.

	\begin{proposition}
		There are $n$-vertex cographs with neighbor-depth $\Omega(\log n)$.
	\end{proposition}
	\begin{proof}
		We will define graphs $G_i$ and $G'_i$ recursively for $i \ge 1$.
		Let $G_1$ be a graph with one vertex.
		For each $i$, the graph $G'_i$ is the disjoint union of two copies of $G_i$.
		For $i \ge 2$, the graph $G_i$ is formed by joining two copies of $G'_{i-1}$ by adding a complete bipartite graph between them.
		The graph $G_i$ is a cograph for all $i$.
		The graph $G_i$ is connected, and for any choice of $v \in V(G_i)$, the graph $G_i \setminus N[v]$ contains the graph $G_{i-1}$ as an induced subgraph.
		Therefore, by induction $G_i$ has neighbor-depth at least $i$.
		The number of vertices of $G_i$ is $4^{i-1}$.
	\end{proof}

	\section{Conclusion}
	We conclude with some open problems.
	First, as already discussed, it is still open if independent set can be solved in polynomial-time on graphs of bounded mim-width, because it is not known how to construct a decomposition of bounded mim-width if one exists.
	It would be very interesting to resolve this problem by either giving an algorithm for computing decompositions of bounded mim-width, or by defining an alternative width parameter that is more general than mim-width and allows to solve \textsc{Independent Set} in polynomial-time when the parameter is bounded.

	The class of graphs of polylogarithmic neighbor-depth generalizes several classes where \textsc{Independent Set} can be solved in (quasi)polynomial time.
	Another interesting class where \textsc{Independent Set} can be solved in polynomial-time and which, to our knowledge, could have polylogarithmic neighbor-depth is the class of graphs with polynomial number of minimal separators~\cite{DBLP:journals/siamcomp/FominTV15}.
	It would be interesting to show that this class has polylogarithmic neighbor-depth.
	More generally, Korhonen~\cite{DBLP:conf/icalp/Korhonen21} studied a specific model of dynamic programming algorithms for \textsc{Independent Set}, in particular, tropical circuits for independent set.
	It appears plausible that all graphs with polynomial size tropical circuits for independent set could have polylogarithmic neighbor-depth.
	
	%
	%
	%
	\bibliographystyle{splncs04}
	\bibliography{biblio}
	
\end{document}